\documentclass[10pt,a4paper,twocolumn]{article}
\usepackage[top=0.8in,bottom=0.8in,left=0.6in,right=0.6in,columnsep=0.25in]{geometry}
\usepackage{amsmath,amssymb,amsthm}
\usepackage[hidelinks]{hyperref}
\usepackage{xcolor}

\newif\ifmarkup \markupfalse
\ifdefined\MarkupOn \markuptrue \fi
\ifmarkup
\newcommand{\rev}[1]{{\color{blue}#1}}
\else
\newcommand{\rev}[1]{#1}
\fi
\ifmarkup
\newcommand{\zw}[1]{{\color{red}#1}}
\else
\newcommand{\zw}[1]{#1}
\fi

\newtheorem{theorem}{Theorem}[section]

\theoremstyle{remark}
\newtheorem{remark}[theorem]{Remark}
\theoremstyle{definition}

\newcommand{\Def}{\mathrm{Def}}
\newcommand{\Ric}{\mathrm{Ric}}
\newcommand{\divg}{\mathrm{div}}

\newcommand{\Lie}{\mathcal{L}}
\newcommand{\R}{\mathbb{R}}
\newcommand{\HH}{\mathbb{H}}
\newcommand{\Sph}{\mathbb{S}}

\title{Resolving the viscosity operator ambiguity on\\
Riemannian manifolds via a kinematic selection principle}
\author{{Zhi-Wei Wang}${}^{1,*}$ and Samuel L.\ Braunstein${}^{2,\dagger}$\\[6pt]
\small ${}^{1}$College of Physics, Jilin University, Changchun 130012, China\\
\small ${}^{2}$Computer Science, University of York, York YO10 5GH, UK\\[6pt]
\small $^*$Corresponding author e-mail: {zhiweiwang.phy@gmail.com} \\
\small $^\dagger$E-mail: {sam.braunstein@york.ac.uk}}
\date{\today}

\begin{document}
\maketitle

\begin{abstract}
On a general Riemannian manifold the Navier-Stokes equations admit
several inequivalent formulations, differing in the choice of viscous
operator: the Hodge Laplacian, the Bochner Laplacian, or the deformation
Laplacian. We \rev{assemble, from classical ingredients,} a Lagrangian
kinematic construction, in which
the strain rate is built from the rate of change of inner products of
Lie-dragged connecting vectors, \rev{and show that it} selects the deformation
Laplacian. The Hodge Laplacian is excluded at the kinematic step (before
any constitutive assumptions) because the strain rate constructed from
inner-product geometry is necessarily symmetric. We further show that
when the fluid arises as a thin-shell limit of \rev{a three-dimensional
flow in a spherical shell}, the operator that emerges depends on the
boundary condition imposed \rev{on the walls of the shell}: stress-free
conditions recover the deformation Laplacian, while \rev{vorticity-free
(Hodge)} boundary
conditions recover the Hodge Laplacian, via an explicit decomposition of
the ambient Bochner Laplacian into intrinsic and extrinsic pieces. The
intrinsic piece is the deformation Laplacian regardless of the boundary
condition. As an analytical \rev{complement}, we prove that on \rev{a}
complete two-dimensional manifold \rev{of bounded geometry} with Gaussian
curvature bounded above by a negative constant, \rev{not necessarily
simply connected,} the incompressible Navier-Stokes equation with the
deformation Laplacian \rev{has} a unique global weak solution \rev{for
initial data in the $L^2$ closure of the compactly supported
divergence-free fields, and that its energy decays exponentially at a
rate set by the curvature bound}.
\end{abstract}

%%%%%%%%%%%%%%%%%%%%%%%%%%%%%%%%%%%%%%%%%%%%%%%%%%
\section{Introduction}
\label{sec:intro}
%%%%%%%%%%%%%%%%%%%%%%%%%%%%%%%%%%%%%%%%%%%%%%%%%%

What should the Navier-Stokes equations look like on a Riemannian
manifold? This question, as Czubak~\cite{Czubak2024} has recently
surveyed, turns out to be difficult. On flat Euclidean
space $\R^n$, the incompressible Navier-Stokes equations take the
familiar form
\begin{equation}\label{eq:NS-flat}
\partial_t u + \nabla_u u + \nabla p = \mu\,\Delta u, \qquad \divg\,u = 0,
\end{equation}
where $u$ is the velocity field, $p$ the pressure, $\mu$ the dynamic
viscosity, and $\Delta = \sum_i \partial_i^2$ is the component-wise
Laplacian. The difficulty is that this component-wise Laplacian is
special to Cartesian coordinates. On a general Riemannian manifold
$(M,g)$, there are at least three natural candidates for a Laplacian
acting on vector fields (or, equivalently, on 1-forms), and they do not
agree when the curvature is non-zero.

The three candidates are the Bochner Laplacian $\Delta_B =
\nabla^i\nabla_i$ (the trace of the second covariant derivative), the
Hodge Laplacian $\Delta_H = -(\mathrm{d}\,\mathrm{d}^* +
\mathrm{d}^*\mathrm{d})$, and the deformation Laplacian $\Delta_\Def =
\Delta_B + \Ric$ (on divergence-free fields). These are related by the
Weitzenb\"ock formula
\begin{equation}\label{eq:weitz}
\Delta_H = \Delta_B - \Ric,
\end{equation}
so on flat space, where $\Ric = 0$, all three agree. On a curved
manifold they differ by multiples of the Ricci curvature, and the
question of which one gives the ``right'' Navier-Stokes equation is a
question about the physical origin of viscosity.

The question has a long history. Ebin and Marsden~\cite{EM70}, building
on Arnold's variational formulation of ideal fluid flow~\cite{Arn66},
used the Hodge Laplacian in their seminal 1970 study of incompressible
flow on manifolds, but added a note in proof stating that the
deformation Laplacian $2\,\divg\,\Def$ is the physically correct
choice. Taylor~\cite{Tay92,TayPDE3} reached the same conclusion via
the continuum-mechanics framework of Serrin~\cite{Ser59}. Despite
these endorsements, the Hodge Laplacian has remained widely used,
particularly on compact manifolds where the Ricci correction might
appear benign.

As Czubak~\cite{Czubak2024} documents, the situation is more tangled
than it first appears. Arguments from continuum mechanics
(Taylor~\cite{Tay92}, Serrin~\cite{Ser59}) point to the deformation
tensor, but rely on phenomenological assumptions (isotropy, spatial
homogeneity) whose status on a general manifold is unclear.
\rev{Several analytical works adopt the deformation operator. Chan,
Czubak and Disconzi~\cite{CCD17} argue for it, in part because on
$\HH^2$ the Hodge form does not control the $H^1$ norm on $L^2$ harmonic
gradients, so that a global energy estimate cannot be derived for the
Hodge operator from a weak formulation that admits them; and
Pierfelice~\cite{Pie17}, following~\cite{EM70,TayPDE3}, develops Stokes
estimates, Fujita-Kato solutions and, in two dimensions, unique Leray
weak solutions for the Navier-Stokes equations written with it, on
simply connected manifolds \zw{of nonpositive sectional curvature with
Ricci curvature pinched between two negative constants}. On the
hyperbolic plane $\HH^2$, $L^2$ harmonic gradients make Leray-Hopf weak
solutions non-unique~\cite{CC10,KM12}, and the weak formulations
of~\cite{CC13b,Pie17} restore uniqueness by selecting among them.}
Geometric
approaches via Gauss formulas~\cite{CC22,CCY23} reveal that projecting
the ambient Laplacian onto a submanifold yields an operator that depends
on the choice of extension of the vector field into the ambient space.
And thin-shell limits~\cite{TZ97,Miu20} give different answers depending
on the boundary condition imposed. Probabilistic approaches via
stochastic Lagrangian paths~\cite{AC12,ACF18,Fan20} recover the
Navier-Stokes equation with either the Hodge or deformation Laplacian
depending on the stochastic formulation, but do not by themselves select
between the two. Czubak concludes that ``the exact form of the equations
will depend on the physical problem at hand.''

In this paper we offer a different perspective. We argue that there is,
in fact, a canonical choice of viscous operator for what one might call
an \emph{intrinsic} fluid on $(M,g)$: one whose strain rate is
determined by the intrinsic geometry of the manifold and the kinematics
of the flow, without reference to any ambient space. The argument
proceeds in the Lagrangian frame, following individual fluid particles
and measuring how the distances between them change over time. The key
observation is that the rate-of-strain tensor constructed this way is
forced, by the geometry of inner products, to be the Lie derivative of
the metric $\Lie_u g = \nabla_a u_b + \nabla_b u_a$. This is twice the
deformation tensor. Its symmetry is not a constitutive assumption but a
consequence of the construction: one is differentiating a symmetric
inner product $g(s,s')$, and no antisymmetric piece can appear. From
this starting point, the standard isotropic Hookean closure and momentum
balance lead uniquely to the deformation Laplacian. \rev{Each ingredient
of this argument is classical (Remark~\ref{rem:prior}); the point of
assembling them is that the Hodge Laplacian is excluded at the kinematic
step, before any constitutive assumption enters.}

We then address the thin-shell case separately, showing by an explicit
decomposition of the ambient 3D Bochner Laplacian on a spherical shell
that the operator splits into an intrinsic kinematic part (the
deformation Laplacian) plus a radial boundary-shear term whose value is
determined by the boundary condition\rev{, imposed on both walls of the
shell}. Stress-free (Navier slip)
conditions set this term to zero, recovering the deformation Laplacian;
Hodge boundary conditions produce a term equal to $-2\,\Ric^{(2)}$,
which shifts the operator to the Hodge Laplacian via the Weitzenb\"ock
identity. \rev{The extension to arbitrary hypersurfaces, a one-parameter
family of wall conditions interpolating between the two cases, and
rigorous thin-shell limits on a class of surfaces of revolution are
given in our paper~\cite{WB-shell}; the spherical computation is kept
here because it is explicit.}

\rev{Section~\ref{sec:analysis} compares the two operators analytically.
On a negatively curved manifold the Ricci term in $\Delta_\Def =
\Delta_B + \Ric$ adds to the dissipation, so the deformation form is
coercive under a pointwise curvature bound alone. For the Hodge
Laplacian the Ricci term enters with the opposite sign; coercivity on
solenoidal fields then needs global input, \zw{topological as well as
spectral,} and it fails on any space of solenoidal fields that contains a nonzero
$L^2$ harmonic field, which is the obstruction found on $\HH^2$
in~\cite{CCD17}.} The \rev{asymptotic results} of Samavaki and
Tuomela~\cite{ST20} on compact manifolds \rev{fit the same picture}: the
kernel of the deformation tensor consists of Killing fields, so viscous
dissipation drives the flow toward rigid motions, which is the
physically expected \rev{behavior}.

The paper is \rev{organized} as follows. In Section~\ref{sec:prelim} we fix
notation and recall the three Laplacians and the Weitzenb\"ock formula.
In Section~\ref{sec:flat} we present the Lagrangian construction in flat
space, which serves both as a pedagogical warmup and as a self-contained
derivation of the Navier-Stokes equations. In Section~\ref{sec:curved}
we carry the construction to a general Riemannian manifold, identifying
the strain rate with $\Lie_u g$ and deriving the deformation Laplacian.
Section~\ref{sec:shell} treats the thin-shell case.
Section~\ref{sec:analysis} discusses the analytical consequences:
coercivity on negatively curved manifolds, \rev{global weak solutions on
negatively curved surfaces,} Korn's inequality, and
long-time \rev{behavior}. Section~\ref{sec:discussion} places the results in
context and discusses extensions to the relativistic setting.
Appendix~\ref{app:tensor} provides a detailed tensor verification of the
thin-shell decomposition.

%%%%%%%%%%%%%%%%%%%%%%%%%%%%%%%%%%%%%%%%%%%%%%%%%%
\section{Preliminaries: Laplacians on vector fields}
\label{sec:prelim}
%%%%%%%%%%%%%%%%%%%%%%%%%%%%%%%%%%%%%%%%%%%%%%%%%%

Let $(M,g)$ be a smooth Riemannian manifold of dimension $d$, and let
$\nabla$ denote its Levi-Civita connection. We consider time-dependent
vector fields $u$ on $M$ satisfying the incompressibility constraint
$\divg\,u \equiv \nabla_a u^a = 0$.

The \emph{Bochner Laplacian} (also called the rough Laplacian or
trace Laplacian) acts on vector fields by
\begin{equation}\label{eq:bochner}
(\Delta_B u)^a = g^{ij}\nabla_i\nabla_j u^a.
\end{equation}
This is the trace of the second covariant derivative. It is the most
direct \rev{generalization} of the component-wise Laplacian, and it reduces to
$\sum_i \partial_i^2 u^a$ in Cartesian coordinates on $\R^d$. On a
compact manifold without boundary, $\Delta_B$ is a negative-semidefinite
operator: $\langle -\Delta_B u, u\rangle = \|\nabla u\|^2 \geq 0$.

The \emph{Hodge Laplacian} acts on differential forms. Identifying
vector fields with 1-forms via the metric (lowering the index), we write
\begin{equation}\label{eq:hodge}
\Delta_H = -(\mathrm{d}\,\mathrm{d}^* + \mathrm{d}^*\mathrm{d}),
\end{equation}
where $\mathrm{d}$ is the exterior derivative and $\mathrm{d}^*$ its
formal adjoint. On divergence-free vector fields, $\mathrm{d}^* u^\flat
= -\divg\,u = 0$, so $\Delta_H$ reduces to $-\mathrm{d}^*\mathrm{d}$. On
$\R^3$, this is the familiar $\nabla(\divg) -
\mathrm{curl}\,\mathrm{curl}$ identity.

The two Laplacians are related by the \emph{Weitzenb\"ock formula}:
\begin{equation}\label{eq:weitz2}
\Delta_H = \Delta_B - \Ric,
\end{equation}
where $\Ric$ denotes the Ricci operator (the Ricci tensor with one index
raised). On flat space $\Ric = 0$ and the two agree. On the round sphere
$\Sph^d$ of radius $R$, $\Ric = (d-1)/R^2$; on the hyperbolic space
$\HH^d$, $\Ric = -(d-1)/R^2$.

The \emph{deformation tensor} of a vector field $u$ is the symmetric
part of its covariant derivative:
\begin{equation}\label{eq:def}
(\Def\,u)_{ab} = \tfrac{1}{2}(\nabla_a u_b + \nabla_b u_a).
\end{equation}
The \emph{deformation Laplacian} is defined by
\begin{equation}\label{eq:defLap}
\Delta_\Def u = 2\,\divg\,\Def(u).
\end{equation}
A direct computation using the Ricci commutator identity
$\nabla^a\nabla_b u_a = \nabla_b(\divg\,u) + R_b{}^c u_c$ shows that,
on divergence-free vector fields,
\begin{equation}\label{eq:defLap2}
\Delta_\Def u = \Delta_B u + \Ric\,u.
\end{equation}
This connects all three operators. On divergence-free fields:
\begin{equation}\label{eq:triangle}
\Delta_\Def = \Delta_B + \Ric = \Delta_H + 2\,\Ric.
\end{equation}

%%%%%%%%%%%%%%%%%%%%%%%%%%%%%%%%%%%%%%%%%%%%%%%%%%
\section{The Lagrangian construction in flat space}
\label{sec:flat}
%%%%%%%%%%%%%%%%%%%%%%%%%%%%%%%%%%%%%%%%%%%%%%%%%%

Before turning to curved manifolds, it is instructive to carry out the
construction in the most transparent setting: flat Euclidean space
$\R^d$. This section serves as both a pedagogical warmup and a
self-contained derivation of the Navier-Stokes equations from the
Lagrangian point of view. For an earlier treatment of the flat-space
Lagrangian derivation with a focus on the stress tensor, see
Deissler~\cite{Dei76}.

\subsection{Following particles}

The Eulerian description of a fluid specifies the velocity $u(t,x)$,
pressure $p(t,x)$, and density $\rho(t,x)$ at each fixed spatial point
$x$. The Lagrangian description instead follows individual fluid
particles (see, e.g., Batchelor~\cite{Bat99} for a thorough treatment).
If a particle is at position $x_0$ at a fiducial time $t_0$, its
position at a later time $t$ is
\begin{equation}\label{eq:flow}
x = x(t,x_0),\qquad x(t_0,x_0) = x_0.
\end{equation}
The velocity of the fluid at position $x(t,x_0)$ is just the velocity
of the particle:
\begin{equation}\label{eq:vel}
u(t,x)\big|_{x=x(t,x_0)} = \frac{\partial x(t,x_0)}{\partial t}.
\end{equation}
Any Lagrangian quantity $\Theta(t,x_0)$ that depends locally on the
particle position can be rewritten as an Eulerian quantity
$\theta(t,x)$ via $\Theta(t,x_0) = \theta(t,x(t,x_0))$. Differentiating
in time gives the convective (material) derivative:
\begin{equation}\label{eq:convective}
\frac{\partial\Theta}{\partial t} = \frac{D\theta}{Dt}
\bigg|_{x=x(t,x_0)},\qquad \frac{D}{Dt} = \partial_t + u\cdot\nabla.
\end{equation}

A key advantage of the Lagrangian frame is that mass conservation
becomes trivial. An elementary volume $\mathrm{d}^d x_0$ of fluid,
carrying mass $\rho(t_0,x_0)\,\mathrm{d}^d x_0$, flows to a new volume
$\mathrm{d}^d x = J(t,x_0)\,\mathrm{d}^d x_0$ at time $t$, where $J$ is
the determinant of the Jacobian $\partial x/\partial x_0$. Since the
same particles occupy both volumes, $\rho(t,x)\,\mathrm{d}^d x$ is
conserved along the flow. It follows that the rate of change of linear
momentum of a fluid element is simply
\begin{equation}\label{eq:momentum}
\frac{D[\rho\,\mathrm{d}^d x\, u]}{Dt} = \rho\,\frac{Du}{Dt}\,\mathrm{d}^d x,
\end{equation}
where $\rho\,\mathrm{d}^d x$ comes out of the derivative precisely
because it is conserved. This is the standard rate-of-change-of-momentum
term in the Navier-Stokes equations; the Lagrangian formulation makes
clear why the density does not appear under the material derivative.

\subsection{Connecting vectors and the strain rate}

Now consider a pair of nearby particles, initially at positions
$x_0$ and $x_0 + s_0$. Their separation vector at time $t$ is
\begin{equation}\label{eq:sep}
s(t,x_0,s_0) = x(t,x_0+s_0) - x(t,x_0).
\end{equation}
Differentiating in time and expanding to first order in $s$:
\begin{equation}\label{eq:dsdt}
\frac{\partial s}{\partial t} = u(t,x+s) - u(t,x)
= (s\cdot\nabla)u(t,x)\big|_{x=x(t,x_0)} + O(s^2).
\end{equation}
This is the equation of geodesic deviation for the flow: nearby
particles separate (or converge) at a rate determined by the
velocity gradient.

Now consider two such separation vectors $s$ and $s'$ emanating from a
common particle. The rate of change of their inner product is
\begin{equation}\label{eq:inner}
\frac{\partial}{\partial t}(s'\cdot s)
= s'\cdot(\nabla\otimes u)\cdot s
+ s\cdot(\nabla\otimes u)\cdot s' + O(s^2 s', s'^2 s).
\end{equation}
In components:
\begin{equation}\label{eq:strain-flat}
\frac{\partial}{\partial t}(s'_a\,\delta_{ab}\,s_b)
= s'_a\,(\partial_a u_b + \partial_b u_a)\,s_b
\equiv s'_a\,\mathfrak{g}_{ab}\,s_b.
\end{equation}
The tensor $\mathfrak{g}_{ab} = \partial_a u_b + \partial_b u_a$ is the
\emph{rate-of-strain tensor}. It encodes all the information about how
local distances and angles between fluid particles change in time. Two
features are manifest. First, $\mathfrak{g}$ is symmetric, because the
inner product $s'\cdot s$ is symmetric in its arguments. Second,
$\mathfrak{g}$ involves only the \emph{symmetric} part of the velocity
gradient; the antisymmetric part $\partial_a u_b - \partial_b u_a$ (the
vorticity) measures local rotation and does not contribute to changes in
inter-particle distances.

\subsection{Constitutive closure}

The physical forces on a fluid element due to its changing strain are
given by the associated stress, through a constitutive relation that
connects the strain rate to the stress. For a \emph{simple} (Newtonian)
fluid, we assume a linear relationship: instantaneous strain rate
produces a linearly proportional instantaneous stress. For an
\emph{isotropic} fluid (one that looks the same in every direction), the
symmetric rate-of-strain tensor $\mathfrak{g}$ decomposes under the
orthogonal group into two irreducible parts: a pure trace (corresponding
to bulk volume changes) and a trace-free symmetric part (corresponding
to pure shear). There can therefore be at most two independent viscosity
coefficients: the bulk viscosity $\zeta$ and the shear viscosity $\mu$.
The resulting stress tensor is
\begin{equation}\label{eq:stress-flat}
\tau_{ab} = \zeta\,(\partial_k u_k)\,\delta_{ab}
+ \mu\left(\partial_a u_b + \partial_b u_a
- \tfrac{2}{d}\,\partial_k u_k\,\delta_{ab}\right).
\end{equation}

\subsection{Momentum balance}

The force on a fluid element $\mathrm{d}^d x$ due to pressure $p$ and
viscous stress $\tau$ is $\nabla\cdot(-p\,\mathbb{I} +
\tau)\,\mathrm{d}^d x$, by Gauss's theorem. Equating this with the rate
of change of momentum from~\eqref{eq:momentum}, and \rev{specializing} to an
incompressible fluid ($\divg\,u = 0$), gives
\begin{equation}\label{eq:NS-flat2}
\rho\,\frac{Du}{Dt} = -\nabla p + \mu\,\Delta u,
\end{equation}
which is the incompressible Navier-Stokes equation~\eqref{eq:NS-flat}.
The Laplacian $\Delta = \sum_i \partial_i^2$ arises from the divergence
of the stress tensor~\eqref{eq:stress-flat} when the fluid is
incompressible.

The entire chain of reasoning is: (i)~follow particles, (ii)~measure how
their separations change (obtaining the symmetric strain rate),
(iii)~close with isotropic Hooke's law, (iv)~take the divergence to get
the force. Each step has a unique outcome, and the result is the
standard equation. The question is what happens when this chain is
repeated on a curved manifold, where the three candidate Laplacians no longer agree.

%%%%%%%%%%%%%%%%%%%%%%%%%%%%%%%%%%%%%%%%%%%%%%%%%%
\section{The kinematic argument on a Riemannian manifold}
\label{sec:curved}
%%%%%%%%%%%%%%%%%%%%%%%%%%%%%%%%%%%%%%%%%%%%%%%%%%

We now carry the Lagrangian construction of Section~\ref{sec:flat} to a
general Riemannian manifold $(M,g)$. Every step has a unique
curved-space replacement; the only new ingredient is the Ricci
curvature, which enters at the very last step when we take the
divergence of the stress.

\subsection{Connecting vectors and Lie transport}

Let $u(t,\cdot)$ be a time-dependent vector field on $M$ generating a
flow $\phi_t$, so that a particle initially at $x_0$ is at $\phi_t(x_0)$
at time $t$. Take a smooth one-parameter family of initial points
$x_0(\tau)$ with $\frac{dx_0}{d\tau}\big|_{\tau=0} = s_0 \in T_{x_0}M$,
and follow each particle along the flow. The connecting vector at time
$t$ is
\begin{equation}\label{eq:connecting}
s(t) = \frac{\partial \phi_t(x_0(\tau))}{\partial\tau}\bigg|_{\tau=0}
\in T_{\phi_t(x_0)}M.
\end{equation}
Here $\tau$ labels a continuous family of distinct fluid particles
whose initial positions form a curve through $x_0$; the derivative
with respect to $\tau$ at $\tau = 0$ extracts the infinitesimal
separation between \rev{neighboring} particles in this family, which is
the manifold analogue of the finite separation vector
$s = x(t,x_0+s_0) - x(t,x_0)$ used in Section~\ref{sec:flat}.
Because $\partial_t$ and $\partial_\tau$ commute on the parameter
rectangle, the connecting vector is \emph{Lie-dragged} by
$u$: $[u,s] = 0$. Using the torsion-free property of the
Levi-Civita connection, this gives
\begin{equation}\label{eq:lie-drag}
\nabla_u s = \nabla_s u.
\end{equation}
This is the curved-space replacement of $\partial_t s = (s\cdot\nabla)u$
from Section~\ref{sec:flat}. The replacement is unique: on a Riemannian
manifold with the Levi-Civita connection, the commutator $[u,s] =
\nabla_u s - \nabla_s u$ (torsion-free), and the Lie-drag condition
$[u,s]=0$ forces~\eqref{eq:lie-drag}.

\subsection{Strain rate as the Lie derivative of the metric}

Let $s$ and $s'$ be two such Lie-dragged connecting vectors emanating
from the same trajectory. Their inner product $g(s,s')$ is a scalar
along the trajectory, and its Lagrangian time derivative is
$u(g(s,s'))$. By metric compatibility ($\nabla g = 0$):
\begin{equation}\label{eq:inner-curved}
\frac{d}{dt}\,g(s,s') = g(\nabla_u s, s') + g(s, \nabla_u s').
\end{equation}
Using the Lie-drag relation $\nabla_u s = \nabla_s u$ and its
primed counterpart:
\begin{equation}\label{eq:strain-curved}
\begin{split}
\frac{d}{dt}\,g(s,s') &= g(\nabla_s u, s') + g(s, \nabla_{s'} u)\\
&= (\nabla_a u_b + \nabla_b u_a)\,s^a(s')^b\\
&= (\Lie_u g)_{ab}\,s^a(s')^b.
\end{split}
\end{equation}
The rate-of-strain tensor on $(M,g)$ is therefore
\begin{equation}\label{eq:strain-Lie}
\mathfrak{g}_{ab} = (\Lie_u g)_{ab} = \nabla_a u_b + \nabla_b u_a
= 2\,(\Def\,u)_{ab}.
\end{equation}

The use of $\mathcal{L}_u g$ as the deformation tensor for viscous
fluids on manifolds goes back to Ebin and
Marsden~\cite{EM70}\rev{, and the operator obtained from it below by
isotropic Hookean closure is the one adopted
in~\cite{EM70,Tay92,TayPDE3,Pie17}. What this section adds is the
observation that the kinematic step alone excludes the Hodge Laplacian
(Remark~\ref{rem:prior}).}

Serrin postulated the deformation tensor (the symmetric part of the
velocity gradient) as the basis for the viscous stress. The
divergence of the resulting stress gives the deformation Laplacian
on any manifold. The Lie derivative $\mathcal{L}_u g$ is the
complete strain rate, not an alternative postulate: $g(s, s')$ is
the full inner product on material connecting vectors $s, s'$
Lie-dragged by the flow, and its time derivative
$(\mathcal{L}_u g)(s, s')$ captures all information about how
lengths and angles of material elements change. \rev{This is the
kinematic justification of} Serrin's postulate\rev{, which is classical
(Remark~\ref{rem:prior})}: the
deformation tensor is the unique tensor measuring the rate of change
of material inner products, derivable from Lagrangian kinematics
without constitutive axioms.

Two features of this result deserve emphasis, because they are the core
of the selection argument.

\emph{Symmetry is forced.} The tensor $\mathfrak{g}$ is symmetric
because $g(s,s')$ is symmetric in its arguments. No antisymmetric piece
(such as the vorticity 2-form $\mathrm{d}u^\flat$, whose components are
$\nabla_a u_b - \nabla_b u_a$) can appear in the strain rate constructed
from inner products. This is not a constitutive assumption: it is a
consequence of the geometry of the construction. This contrasts with the
Hodge Laplacian, which on divergence-free fields reduces to
$-\mathrm{d}^*\mathrm{d}$, where $\mathrm{d}$ produces the antisymmetric
part of $\nabla u$. That antisymmetric part is vorticity, which is a
kinematic invariant of the flow but is not a strain rate.

\emph{Intrinsicality is automatic.} The construction uses only the
metric $g$, the Levi-Civita connection, and the flow. No embedding,
extension, or ambient space is involved. This is in contrast with the
Gauss-formula approach~\cite{CC22,CCY23}, where the projected Laplacian
depends on how the vector field is extended off the submanifold.

\subsection{Constitutive closure and the deformation Laplacian}

The constitutive step is identical to the flat-space case. Linearity and
isotropy give two viscosity coefficients (bulk $\zeta$ and shear $\mu$),
and the stress tensor is
\begin{equation}\label{eq:stress-curved}
\tau_{ab} = \zeta\,(\nabla\!\cdot\!u)\,g_{ab} + \mu\left[\nabla_a u_b
+ \nabla_b u_a - \tfrac{2}{d}\,(\nabla\!\cdot\!u)\,g_{ab}\right].
\end{equation}
Note that $g_{ab}$ replaces $\delta_{ab}$; this is the only change
from~\eqref{eq:stress-flat}. Momentum balance gives $\rho\,\frac{Du}{Dt}
= -\nabla p + \nabla^a\tau_{ab}$, where the divergence of the stress is
computed using the Levi-Civita connection.

Taking the divergence:
\begin{equation}
\begin{split}
\nabla^a \tau_{ab} &= \zeta\,\nabla_b(\nabla\!\cdot\!u)\\
&\quad + \mu\left[\nabla^a\nabla_a u_b + \nabla^a\nabla_b u_a
- \tfrac{2}{d}\,\nabla_b(\nabla\!\cdot\!u)\right].
\end{split}
\end{equation}
The first term is the Bochner Laplacian $\Delta_B u$. For the second,
we use the standard commutator identity for vector fields:
\begin{equation}\label{eq:commutator}
\nabla^a\nabla_b u_a = \nabla_b(\nabla\!\cdot\!u) + R_b{}^c u_c,
\end{equation}
where $R_b{}^c$ is the Ricci tensor with one index raised. For
an incompressible fluid ($\nabla\!\cdot\!u = 0$), this yields
\begin{equation}\label{eq:NS-curved}
\rho\,(\partial_t u + \nabla_u u) + \nabla p
= \mu\,(\Delta_B u + \Ric\,u) = \mu\,\Delta_\Def u,
\end{equation}
which is the deformation Laplacian. The Ricci curvature enters at
exactly one place: the commutator~\eqref{eq:commutator}, which arises
when we take the divergence of the \rev{symmetrized} covariant derivative.
This is the only step in the entire construction that is sensitive to
the curvature of $(M,g)$.

\begin{remark}[Material Frame-Indifference] The kinematic construction
automatically satisfies the Principle of Material Frame-Indifference:
the strain rate $\mathfrak{g} = \mathcal{L}_u g$ is invariant under
superposed rigid-body motions (isometries of $(M,g)$), because Killing
fields $\xi$ satisfy $\mathcal{L}_\xi g = 0$ by definition, so replacing
$u$ by $u + \xi$ shifts $\mathfrak{g}$ by $\mathcal{L}_\xi g = 0$. This
provides an independent confirmation that the deformation tensor, and
not the full velocity gradient, is the appropriate argument of the
constitutive law. The vorticity $\mathrm{d}u^\flat$ does not have this
property: it is not invariant under superposed rotations. \end{remark}

\begin{remark}\label{rem:prior}
The kinematic identification $\mathfrak{g} = \Lie_u g$ is not new.
It appears as Proposition~3.30 and Corollary~6.12(i) in Marsden and
Hughes~\cite{MH83}, and the broader continuum-mechanics and geometric
mechanics traditions (Truesdell~\cite{Tru55}, Oldroyd~\cite{Old50},
Holm, Marsden and Ratiu~\cite{HMR98}) treat the \rev{symmetrized} covariant
derivative as the natural strain rate. The isotropic Hookean
closure is standard. The Ricci commutator step is in
Taylor~\cite{Tay92,TayPDE3}\rev{, and the resulting operator is the one
used in analytical work on non-compact manifolds~\cite{Pie17}. Our
purpose in this section is to assemble these ingredients into a
selection argument addressed to} the question of~\cite{Czubak2024}:
\rev{the} Hodge Laplacian is
excluded at the kinematic step, before any constitutive assumptions
enter, because the strain rate is forced to be symmetric by the geometry
of inner products.
\end{remark}

%%%%%%%%%%%%%%%%%%%%%%%%%%%%%%%%%%%%%%%%%%%%%%%%%%
\section{The thin-shell case}
\label{sec:shell}
%%%%%%%%%%%%%%%%%%%%%%%%%%%%%%%%%%%%%%%%%%%%%%%%%%

The kinematic argument of Section~\ref{sec:curved} applies to fluids
whose configuration space is intrinsically the manifold $(M,g)$. A
different situation arises when the fluid on $M$ is a \emph{thin-shell
limit} of a 3D fluid in an ambient space. Czubak~\cite{Czubak2024}
surveys two such constructions: Temam and Ziane~\cite{TZ97}, who impose
Hodge boundary conditions and obtain the Hodge Laplacian on $\Sph^2$;
and Miura~\cite{Miu20}, who \rev{imposes} Navier slip conditions and obtains
the deformation Laplacian. Both results are correct. The question is why
they disagree.

In this section we answer this question by decomposing the ambient 3D
Bochner Laplacian on a thin spherical shell into an intrinsic piece and
a radial boundary-shear piece. The intrinsic piece is the deformation
Laplacian, regardless of the boundary condition. The boundary condition
determines the radial piece, and the two standard boundary conditions
give the two textbook operators.

\rev{The radial piece involves the normal profile of the velocity to
second order, through $\partial_r U$ and $\partial_r^2 U$ on the
mid-surface. A wall condition imposed on a single surface fixes
$\partial_r U$ there but leaves $\partial_r^2 U$ free; it is the
condition imposed on both walls of a shell of thickness $\varepsilon$
that determines $\partial_r^2 U$ as $\varepsilon\to 0$. We impose it in
that form, working formally with an expansion of the profile in the
distance from the mid-surface. The same two-wall analysis on arbitrary
hypersurfaces, a one-parameter family of wall conditions interpolating
between the two cases, and the rigorous limits, as Mosco convergence of
the viscous energy forms on a class of surfaces of revolution, are given
in~\cite{WB-shell}.}

\subsection{Setup}

Consider the flat ambient space $\R^3$ with spherical coordinates
$(r,x^1,x^2)$, where $x^j$ are angular coordinates on $\Sph^2$. The
flat metric is
\begin{equation}
\mathrm{d}s^2 = \mathrm{d}r^2 + r^2\tilde{g}_{jk}\,
\mathrm{d}x^j\mathrm{d}x^k,
\end{equation}
where $\tilde{g}_{jk}$ is the standard round metric on the unit
$\Sph^2$. Let $U$ be a purely tangential vector field ($U^r = 0$) \rev{in
the shell $|r-R| < \varepsilon/2$ around the sphere of radius $R$, with
walls at $r = r_\pm = R\pm\varepsilon/2$. We
expand
$$
U^i = q_0^i + (r-R)\,q_1^i + (r-R)^2\,q_2^i + O\bigl((r-R)^3\bigr),
$$
with coefficients $q_k^i$ depending on the angular coordinates but not
on $\varepsilon$, so that $\partial_r U^i = q_1^i$ and
$\partial_r^2 U^i = 2q_2^i$ at $r = R$.}

Because $\R^3$ is flat, the ambient viscous operator is the 3D Bochner
Laplacian $\Delta_B^{(3)}$. We compute its tangential components at $r =
R$ in Appendix~\ref{app:tensor}; the result is:
\begin{equation}\label{eq:decomp}
\begin{split}
(\Delta_B^{(3)} U)^i_{\text{tan}} &= \underbrace{\bigl(\Delta_B^{(2)} U^i
+ \Ric^{(2)}(U)^i\bigr)}_{\text{intrinsic: deformation Laplacian}}\\
&\quad + \underbrace{\bigl(\partial_r^2 U^i
+ \tfrac{4}{R}\,\partial_r U^i\bigr)}_{\text{radial boundary shear, }
\mathcal{F}_\text{rad}}.
\end{split}
\end{equation}
The intrinsic piece is $\Delta_B^{(2)} + \Ric^{(2)} = \Delta_\Def^{(2)}$.
The radial piece $\mathcal{F}_\text{rad}$ \rev{is fixed by} the
boundary condition.

\subsection{Case 1: Navier slip (stress-free) boundary condition}

Zero normal shear requires the tangential deformation tensor components
to vanish at the boundary: $(\Def\,U)_{ri} = 0$. As computed in
Appendix~\ref{app:tensor}, this condition reduces to $\partial_r U^i =
0$ \rev{on each wall. On the walls $\partial_r U^i = q_1^i +
2(r-R)\,q_2^i + O(\varepsilon^2)$, and imposing the condition at
$r = r_\pm$ gives, from the sum and from the difference divided by
$\varepsilon$, $q_1^i = 0$ and $q_2^i = 0$ to leading order. Hence
$\partial_r U^i = \partial_r^2 U^i = 0$ at $r = R$} and
$\mathcal{F}_\text{rad} = 0$, so
\begin{equation}
(\Delta_B^{(3)} U)_\text{tan} = \Delta_B^{(2)} U + \Ric^{(2)}(U)
= \Delta_\Def^{(2)}\,U.
\end{equation}
This recovers Miura's result~\cite{Miu20} and matches the intrinsic
kinematic answer of Section~\ref{sec:curved}.

\subsection{Case 2: Hodge boundary condition}

The Hodge condition requires vanishing tangential vorticity at the
boundary: $[\mathrm{curl}\,U \times \hat{n}]_\text{tan} = 0$. \rev{On a
wall $r = r_\pm$, where $U_r = 0$ and hence $\partial_i U_r = 0$, this}
forces the \rev{radial derivative of the} \emph{covariant} tangential
components to \rev{vanish}: $\partial_r U_i = 0$. Because the metric
expands as $r^2$, \rev{$U_i = r^2\tilde{g}_{ij}U^j$, and the condition on
the wall reads $\partial_r U^i = -(2/r_\pm)\,U^i$. Expanding $-2/r_\pm =
-(2/R)\bigl(1 \mp \varepsilon/(2R)\bigr) + O(\varepsilon^2)$ and imposing
the condition on both walls, half the sum gives
$q_1^i = -(2/R)\,q_0^i$, and the difference divided by $\varepsilon$
gives $2q_2^i = (2/R^2)\,q_0^i - (2/R)\,q_1^i = (6/R^2)\,q_0^i$, to
leading order. Hence, at $r = R$,}
\begin{equation}
\partial_r U^i = -\frac{2}{R}\,U^i,\qquad \partial_r^2 U^i
= \frac{6}{R^2}\,U^i.
\end{equation}
Substituting into $\mathcal{F}_\text{rad}$:
\begin{equation}
\begin{split}
\mathcal{F}_\text{rad} &= \frac{6}{R^2}\,U^i
+ \frac{4}{R}\left(-\frac{2}{R}\,U^i\right)\\
&= -\frac{2}{R^2}\,U^i = -2\,\Ric^{(2)}(U)^i,
\end{split}
\end{equation}
since $\Ric^{(2)} = R^{-2}$ on the round sphere. The total
tangential projection is then
\begin{equation}
\begin{split}
(\Delta_B^{(3)} U)_\text{tan} &= \Delta_B^{(2)} U + \Ric^{(2)}(U)
- 2\,\Ric^{(2)}(U)\\
&= \Delta_B^{(2)} U - \Ric^{(2)}(U) = \Delta_H^{(2)}\,U,
\end{split}
\end{equation}
the Hodge Laplacian via the Weitzenb\"ock identity. This recovers
Temam-Ziane~\cite{TZ97}.

\subsection{Interpretation}

The decomposition~\eqref{eq:decomp} makes the situation transparent. The
intrinsic kinematic part of the ambient viscous operator is always the
deformation Laplacian. What the boundary condition does is inject (or
not) an additional effective drag from the third dimension. The Hodge
boundary condition, \rev{which makes the radial derivative of the
covariant components vanish at both walls}, forces the fluid to shear against
the expanding metric geometry;
this produces a drag of exactly $-2\,\Ric^{(2)}$, which shifts the
operator by the Weitzenb\"ock gap from the deformation to the Hodge
Laplacian. The Navier-slip condition, \rev{which makes the contravariant
components constant across the shell to second order}, imposes no such
drag.

The physical picture is: a fluid that is free in the ambient direction
experiences the intrinsic (deformation) operator. A fluid that is pinned
to the ambient geometry through its boundary experiences an effective
extrinsic friction that is geometric in origin, proportional to the
Ricci curvature. The two standard choices of operator correspond to two
natural physical limits of a single family\rev{; the family is
constructed explicitly in~\cite{WB-shell}}.

%%%%%%%%%%%%%%%%%%%%%%%%%%%%%%%%%%%%%%%%%%%%%%%%%%
\section{Analytical consequences}
\label{sec:analysis}
%%%%%%%%%%%%%%%%%%%%%%%%%%%%%%%%%%%%%%%%%%%%%%%%%%

The kinematic selection of the deformation Laplacian is consistent with
several analytical results in the literature.

\subsection[Coercivity and the Ricci term]{\rev{Coercivity and the Ricci term}}
\label{sec:coercivity}

The energy dissipation rate for a divergence-free velocity field $u$ is
determined by the quadratic form $\langle -Lu, u\rangle$, where $L$ is
the viscous operator. For the two candidates \rev{(for compactly
supported $u$, after integration by parts)}:
\begin{align}
\langle -\Delta_\Def u, u\rangle
&= 2\|\Def\,u\|^2 = \|\nabla u\|^2 - \langle\Ric\,u, u\rangle,
\label{eq:coerce-def}\\
\langle -\Delta_H u, u\rangle
&= \|\mathrm{d}u^\flat\|^2 = \|\nabla u\|^2 + \langle\Ric\,u,
u\rangle.\label{eq:coerce-hodge}
\end{align}
On $\HH^2$, the Ricci curvature acts as $\Ric(u) = -u$ (for $\HH^2$ of
sectional curvature $-1$). So:
\begin{align}
\langle -\Delta_\Def u, u\rangle
&= \|\nabla u\|^2 + \|u\|^2,\\
\langle -\Delta_H u, u\rangle
&= \|\nabla u\|^2 - \|u\|^2 = \|\mathrm{d}u^\flat\|^2 \geq 0.
\end{align}
\rev{The deformation form controls the $H^1$ norm $\|\nabla u\|^2 +
\|u\|^2$ through this identity alone:} the Ricci term \emph{reinforces}
dissipation\rev{, and} negative curvature makes the fluid dissipate more,
not less. The Hodge form is positive semi-definite (being the squared
norm of a differential form), but \rev{the identity by itself does not
bound $\|u\|^2$, and the form vanishes on every $L^2$ harmonic field
$\eta$, since $\mathrm{d}\eta^\flat = 0$ and $\mathrm{d}^*\eta^\flat = 0$.
On $\HH^2$ there are many such fields, $\eta = \nabla\Phi$ with $\Phi$
harmonic and $\nabla\Phi\in L^2$, and on them
$\|\nabla\eta\| = \|\eta\| > 0$ while $\|\mathrm{d}\eta^\flat\| = 0$.
This is the obstruction found by Chan,
Czubak and Disconzi~\cite{CCD17}, who use such a field to show that, for
the Stokes equation with the Hodge Laplacian on $\HH^2$, a global energy
estimate controlling $\|\nabla u\|$ cannot be derived from a weak
formulation that admits it. These fields are divergence-free, but
integration by parts shows that they are orthogonal to every compactly
supported divergence-free field, so they do not lie in the $L^2$ closure
$H$ of those fields, the space of initial data in
Section~\ref{sec:weak-solutions}.

\zw{On $\HH^2$ the Hodge form is nevertheless coercive on $H$}, by a global argument. On a
simply connected non-compact surface, which is diffeomorphic to $\R^2$,
a compactly supported divergence-free field has
$u^\flat = \star\,\mathrm{d}\psi$ with $\psi$ compactly supported,
$\star$ being the Hodge star. Then $\|u\|^2 = \|\nabla\psi\|^2 =
\langle-\Delta\psi,\psi\rangle$ and, by the spectral theorem,
$\|\mathrm{d}u^\flat\|^2 = \|\Delta\psi\|^2 \geq
\lambda_0\|\nabla\psi\|^2$, where $\lambda_0$ is the bottom of the
spectrum of the scalar Laplacian. Hence
\begin{equation}\label{eq:hodge-H}
\|\mathrm{d}u^\flat\|^2 \geq \lambda_0\,\|u\|^2 .
\end{equation}
On $\HH^2$, $\lambda_0 = 1/4$, and then $\|\nabla u\|^2 =
\|\mathrm{d}u^\flat\|^2 + \|u\|^2 \leq 5\,\|\mathrm{d}u^\flat\|^2$. The
difference between the two operators on $\HH^2$ is therefore not whether
they are coercive on $H$, but what the coercivity rests on: a pointwise
curvature bound for the deformation form, and simple connectivity
together with a spectral bound for the Hodge form.
Remark~\ref{rem:hodge-H} gives a surface on which the latter fails.}

The pattern is general. On any manifold with $\Ric \leq 0$, the
deformation Laplacian is \rev{at least as dissipative as} the Bochner
Laplacian, while the Hodge Laplacian is \rev{at most as dissipative} and
may fail to be coercive. On \rev{compact} positively curved manifolds
the roles reverse, \rev{and compactness then supplies coercivity modulo
finite-dimensional kernels (Section~\ref{sec:korn}). In the kinematic
picture, the sign of the Ricci term reflects what viscosity responds
to}: the symmetric strain rate $\Lie_u g = \nabla u + \nabla u^T$, not
the antisymmetric vorticity $\nabla u - \nabla u^T$.

\subsection{Global weak solutions on negatively curved 2D manifolds}
\label{sec:weak-solutions}

\rev{Global weak solutions for the deformation operator on negatively
curved non-compact manifolds are known.
Pierfelice~\cite[Theorem~6.2]{Pie17} proves that on a complete, simply
connected, non-compact surface with curvature and its first two
derivatives bounded, Ricci curvature pinched between two negative
constants, and injectivity radius bounded below, there is a unique Leray
weak solution for every divergence-free initial velocity in $L^2$, the
pressure being required to be the $L^2$ solution of its elliptic
equation; in all dimensions she also develops Stokes estimates and
Fujita-Kato theory. We record here a two-dimensional version that does
not assume simple connectivity, and in which the coercivity comes from
the pointwise curvature bound through~\eqref{eq:coerce-def}, with no
spectral input. The price is a smaller class of initial data, discussed
after the theorem.}

Let $(M,g)$ be a complete, non-compact 2D Riemannian manifold with
bounded geometry (i.e., injectivity radius bounded below and curvature
bounded along with its derivatives). Assume its Gaussian curvature $K$
is strictly negative and bounded away from zero: $K \leq -\kappa^2 < 0$
for some constant $\kappa > 0$. Let $\mathcal{V}$ be the space of
smooth, compactly supported, divergence-free vector fields on $M$. We
define $H$ and $V$ as the closures of $\mathcal{V}$ in the $L^2(M)$ and
$H^1(M)$ norms, respectively. \rev{$M$ need not be simply connected.}

\begin{theorem}\label{thm:weak-M}
Let $\mu > 0$. For any divergence-free initial velocity field
$u_0 \in H$, the incompressible Navier-Stokes equations with
the deformation Laplacian on $M$,
\begin{equation}
\partial_t u + \nabla_u u + \nabla p = \mu\,\Delta_\Def u,
\qquad \divg\,u = 0,
\end{equation}
admit a unique global weak solution $u \in L^\infty(0,T; H) \cap L^2(0,T; V)$
for any $T > 0$, together with a pressure distribution
$p \in \mathcal{D}'((0,T)\times M)$ recovering the momentum equation.
Furthermore, \rev{the energy decays exponentially}:
\begin{equation}
\|u(t)\|_{L^2}^2 \leq e^{\rev{-4\mu\kappa^2 t}} \|u_0\|_{L^2}^2.
\end{equation}
\end{theorem}

Chan and Czubak~\cite{CC10} showed that Leray-Hopf weak solutions
on $\HH^2$ are non-unique even with the deformation Laplacian. Khesin
and Misio{\l}ek~\cite{KM12} showed that this non-uniqueness is a
consequence of the Hodge decomposition specific to dimension two and
does not occur on $\HH^n$ for $n \geq 3$. In three dimensions,
Lichtenfelz~\cite{Lic15} proved non-uniqueness \rev{of Leray-Hopf
solutions} on a negatively curved 3-manifold \rev{carrying $L^2$
harmonic 1-forms}, but only for the Hodge Laplacian\rev{. Harmonic
1-forms drop out of the Hodge viscous term, but not out of the
deformation one, since $\Delta_\Def\eta = 2\Ric(\eta)$ for a harmonic
field $\eta$, which is not a gradient in general.}

\rev{The non-unique solutions on $\HH^2$ have the form
$u = f(t)\nabla\Phi$, with $\Phi$ harmonic, $\nabla\Phi\in L^2$ and $f$
arbitrary~\cite{KM12,Pie17}. As noted in Section~\ref{sec:coercivity},
such fields are orthogonal to $H$, so they are not admissible data for
Theorem~\ref{thm:weak-M}, and its uniqueness statement does not conflict
with these examples. The formulations of~\cite{CC13b} and~\cite{Pie17}
admit them and restore uniqueness by a selection\zw{. For such $u$
(unit viscosity, curvature $-1$) the pressure is, up to a function of
$t$, $p = -(f'+2f)\,\Phi - \tfrac12 f^2|\nabla\Phi|^2$~\cite[Section~6]{Pie17}.
Since $\nabla\Phi$ is bounded as well as square integrable, by elliptic
estimates, $|\nabla\Phi|^2\in L^2$, while no nonzero harmonic function
on $\HH^2$ lies in $L^2$; so} Pierfelice's
requirement that the pressure lie in $L^2(0,T;L^2(M))$ forces
\zw{$f' = -2f$, that is} $f(t) = f(0)e^{-2t}$, and she notes
that the formulation of~\cite{CC13b} leads to the same
selection~\cite[Section~6]{Pie17}. By de Rham's theorem, an $L^2$ field
orthogonal to $H$ is a gradient, so the divergence-free fields in $L^2$
split into $H$ and the $L^2$ harmonic gradients. On a simply connected
surface every $L^2$ harmonic field is such a gradient, and $H$ is the
orthogonal complement of the $L^2$ harmonic fields among the
divergence-free fields in $L^2$; there~\cite[Theorem~6.2]{Pie17}, which
takes all divergence-free data in $L^2$, applies to a class containing
$H$. What Theorem~\ref{thm:weak-M} adds is the case in which $M$ is not
simply connected, where $H$ can contain $L^2$ harmonic fields, and a
decay rate that is attained; Remark~\ref{rem:hodge-H} gives an example,
on which the decay also separates the two operators.}

\begin{proof}
The proof proceeds via the Galerkin method and exhaustion by compact
domains (a standard framework for noncompact domains; see, e.g.,
Heywood~\cite{Hey80} for the $\R^2$ analogue), relying on
the coercivity \rev{bound~\eqref{eq:exact-korn}} of the deformation
operator.

\emph{Step 1: \rev{Coercivity and energy decay}.} As
derived in~\eqref{eq:coerce-def}, for any vector field
$u \in \mathcal{V}$, integration by parts yields the viscous dissipation form:
\begin{equation}
\langle -\Delta_\Def u, u\rangle = 2\|\Def\,u\|_{L^2}^2
= \|\nabla u\|_{L^2}^2 - \langle\Ric\,u, u\rangle.
\end{equation}
On the 2D manifold $M$, the Ricci tensor acts as $\Ric(u) = K u$.
Since $K \leq -\kappa^2 < 0$, we obtain the coercivity bound:
\begin{equation}\label{eq:exact-korn}
\begin{split}
\langle -\Delta_\Def u, u\rangle &= \|\nabla u\|_{L^2}^2
- \int_M K |u|^2\,\mathrm{d}V\\
&\geq \|\nabla u\|_{L^2}^2
+ \kappa^2\|u\|_{L^2}^2 \geq \min(1, \kappa^2) \|u\|_V^2.
\end{split}
\end{equation}
The deformation form is therefore coercive \rev{on $V$. On this class of
manifolds the identity~\eqref{eq:coerce-def} takes the place of a Korn
inequality, with constants fixed by the curvature bound.}

Testing the momentum equation with $u$ (rigorously justified via
Steklov averages for the weak solution) gives the energy equality:
\begin{equation}
\frac{1}{2}\frac{d}{dt}\|u(t)\|_{L^2}^2 + \mu \langle
-\Delta_\Def u, u\rangle = 0.
\end{equation}
\rev{For the decay we use the second form of the dissipation. By
\eqref{eq:coerce-hodge} and~\eqref{eq:triangle}, on divergence-free
fields
\begin{equation}\label{eq:def-hodge-form}
\langle -\Delta_\Def u, u\rangle = \|\mathrm{d}u^\flat\|_{L^2}^2
- 2\int_M K|u|^2\,\mathrm{d}V \geq 2\kappa^2\|u\|_{L^2}^2,
\end{equation}
which holds on $V$ by density. Substituting this into the energy
equality,} we obtain the differential inequality:
\begin{equation}
\frac{d}{dt}\|u(t)\|_{L^2}^2 \leq \rev{-4\mu\kappa^2} \|u(t)\|_{L^2}^2.
\end{equation}
\rev{Gr\"onwall's lemma then gives}
\begin{equation}
\|u(t)\|_{L^2}^2 \leq e^{\rev{-4\mu\kappa^2 t}} \|u_0\|_{L^2}^2.
\end{equation}

\emph{Step 2: Exhaustion by compact domains.} \zw{Let $\gamma$ be a
smooth proper function on $M$, and let $R\to\infty$ run through regular
values of $\gamma$ (almost every value is regular, by Sard's theorem).
The sets $\Omega_R = \{\gamma < R\}$ are relatively compact with smooth
boundary, and they exhaust $M$. Geodesic balls are not used, because
when $M$ is not simply connected their boundaries need not be smooth
where they meet the cut locus.} For each $R$,
let $H_R$ and $V_R$ be the closures of smooth, compactly supported,
divergence-free fields on $\Omega_R$ in $L^2$ and $H^1$, respectively.
We select initial data $u_{0,R} \in H_R$ such that $u_{0,R} \to u_0$ in
$H$ and $\|u_{0,R}\|_{L^2} \leq \|u_0\|_H$. Since $\Omega_R$ is a
bounded domain with a smooth boundary, classical theory guarantees the
existence of a global weak solution $u_R \in L^\infty(0,T; H_R) \cap
L^2(0,T; V_R)$ satisfying Dirichlet boundary conditions on
$\partial\Omega_R$. 

The localized energy inequality on $\Omega_R$ reads:
\begin{equation}
\begin{split}
&\frac{1}{2}\|u_R(t)\|_{L^2(\Omega_R)}^2 + \mu \int_0^t \langle
-\Delta_\Def u_R, u_R \rangle\,\mathrm{d}s\\
&\qquad\leq
\frac{1}{2}\|u_{0,R}\|_{L^2(\Omega_R)}^2 \leq \frac{1}{2}\|u_0\|_H^2.
\end{split}
\end{equation}
Using the coercivity bound~\eqref{eq:exact-korn} on $\Omega_R$
(which holds identically because the boundary terms vanish for $u_R \in
V_R$) and extending $u_R$ by zero to all of $M$, we deduce that the
sequence $\{u_R\}$ is uniformly bounded in $L^\infty(0,T; H) \cap
L^2(0,T; V)$.

\emph{Step 3: Nonlinear bounds and local compactness.} Because $M$ is a
two-dimensional manifold with bounded geometry, the Ladyzhenskaya
inequality $\|v\|_{L^4}^2 \leq C \|v\|_{L^2} \|v\|_V$ holds
globally~\cite{Heb96,Aub98}. The uniform energy bounds then imply that
$\{u_R\}$ (extended by zero to $M$) is uniformly bounded in $L^4(0,T;
L^4(M))$. Hence, the convective term $u_R \otimes u_R$ is uniformly
bounded in $L^2(0,T; L^2(M))$.

By the Banach-Alaoglu theorem, we extract a subsequence (still denoted
$u_R$) such that $u_R \rightharpoonup u$ weakly-* in $L^\infty(0,T; H)$
and weakly in $L^2(0,T; V)$. To pass to the limit in the nonlinear term,
we require strong local convergence. A direct application of the
standard Aubin-Lions-Simon lemma via bounding $\partial_t u_R$ in
$V_W^*$ (for a local domain $W \subset M$) fails here: the natural
restriction mapping $L^2(W) \to V_W^*$ has an infinite-dimensional
kernel of square-integrable harmonic gradients, destroying the strict
injectivity required for the lemma.

Instead, we establish strong local compactness via global fractional
time-translation bounds. Let $V_R^*$ be the dual space of $V_R$.
Evaluating the momentum equation on $\Omega_R$ against any test function
in $V_R$, the uniform bounds on the viscous and convective terms show
that the time derivative satisfies $\|\partial_t u_R\|_{V_R^*} \leq
g_R(t)$, where $g_R$ is uniformly bounded in $L^2(0,T)$ independently of
$R$. Since $u_R(t+h) - u_R(t) \in V_R$, we can test the time derivative
against this difference. A standard time-translation estimate (see e.g.,
Temam~\cite{Tem01}) yields uniform equicontinuity in $L^2(M)$:
\begin{equation}
\int_0^{T-h} \|u_R(t+h) - u_R(t)\|_{L^2(M)}^2 \,\mathrm{d}t \leq C h^{1/2},
\end{equation}
where $C$ is a constant independent of $R$ and $h$.

Now let $W \subset M$ be an arbitrary compact domain. The restriction of
$\{u_R\}$ to $W$ is uniformly bounded in $L^2(0,T; H^1(W))$. By the
Rellich-Kondrachov theorem, the spatial embedding $H^1(W)
\hookrightarrow L^2(W)$ is compact. Combined with the uniform
time-equicontinuity in $L^2(W)$, Simon's compactness
theorem~\cite{Simon86} (or the Riesz-Fr\'echet-Kolmogorov theorem)
guarantees that $\{u_R|_W\}$ is strongly precompact in $L^2(0,T;
L^2(W))$. We thus extract a further diagonal subsequence converging
strongly in $L^2(0,T; L^2_{\text{loc}}(M))$.

\emph{Step 4: Passage to the limit and pressure recovery.} Let $\phi \in
C^\infty_c([0,T) \times M, TM)$ be a smooth divergence-free test
function. For sufficiently large $R$, the spatial support of $\phi$ is
contained in $\Omega_R$, and the weak formulation for $u_R$ gives:
\begin{equation}
\begin{split}
\int_0^T \Bigl( &-\langle u_R, \partial_t \phi\rangle - \langle u_R
\otimes u_R, \nabla \phi\rangle\\
&+ 2\mu\langle \Def\,u_R,
\Def\,\phi\rangle \Bigr)\, \mathrm{d}t = \langle u_{0,R}, \phi(0)\rangle.
\end{split}
\end{equation}
The strong local convergence in $L^2$ and the uniform $L^4$ bound ensure
that $u_R \otimes u_R \rightharpoonup u \otimes u$ when tested against
smooth compactly supported functions. Thus, passing to the limit as $R
\to \infty$ is justified for all terms, demonstrating that $u$ is a
valid global weak solution on $M$.

To resolve the pressure recovery problem formally, consider the residual
distribution $F = \partial_t u + \nabla_u u - \mu \Delta_\Def u \in
\mathcal{D}'((0,T)\times M, TM)$. Since $u$ is a weak solution, $F$
annihilates any smooth test function $\phi$ that is divergence-free and
compactly supported. By De Rham's theorem on manifolds (see e.g.,
\cite{TayPDE3}), any distribution orthogonal to all compactly supported
divergence-free vector fields is exactly the gradient of a scalar
distribution. Hence, there exists a global pressure distribution $p \in
\mathcal{D}'((0,T)\times M)$ such that $F = -\nabla p$, recovering the
pressure field and the full equations.
\end{proof}

\begin{remark}
The solution constructed above is unique \rev{in the class
$L^\infty(0,T;H)\cap L^2(0,T;V)$}. The standard two-dimensional
argument \cite{Tem01} applies: if $u$ and $v$ are two weak solutions
with the same initial data, the difference $w = u - v$ satisfies an
energy inequality. By the Ladyzhenskaya and Young
inequalities, the nonlinear term can be bounded as
$$
\begin{aligned}
|\langle(w\cdot\nabla)u,w\rangle|
&\leq \frac{\mu\min(1,\kappa^2)}{2}\|w\|_V^2\\
&\quad + C(\mu,\kappa)\|\nabla u\|_{L^2}^2\|w\|_{L^2}^2.
\end{aligned}
$$
Absorbing the first term into the dissipative term
$\mu \min(1, \kappa^2)\|w\|_V^2$ and applying Gr\"onwall's lemma to
the second yields $w = 0$.
\end{remark}

\begin{remark}\label{rem:hodge-H}
\rev{The decay rate in Theorem~\ref{thm:weak-M} is attained. Let
\zw{$\mathcal{C}$} be a hyperbolic cylinder, the quotient of $\HH^2$ by a hyperbolic isometry,
with metric $\mathrm{d}y^2 + \cosh^2\!y\,\mathrm{d}z^2$, $y\in\R$,
$z\in\R/\ell\mathbb{Z}$. It is complete, with $K = -1$ and
injectivity radius at least $\ell/2$, so it satisfies the hypotheses of
Theorem~\ref{thm:weak-M} with $\kappa = 1$. Let $\eta$ be the vector
field with $\eta^\flat = \mathrm{d}z$. It is closed and co-closed, with
$|\eta| = \mathrm{sech}\,y$ and $\|\eta\|^2 =
\ell\int_\R\mathrm{sech}\,y\,\mathrm{d}y = \pi\ell$; its covariant
derivative, $\nabla\eta^\flat = -\tanh y\,(\mathrm{d}y\otimes\mathrm{d}z
+ \mathrm{d}z\otimes\mathrm{d}y)$, is square integrable; and
$\chi(y)\,\eta$ is divergence-free for every compactly supported $\chi$,
so that, taking cutoffs $\chi$ increasing to $1$, $\eta\in V$. It is not
a gradient, since the integral of $\mathrm{d}z$ around the closed
geodesic $y = 0$ is $\ell$. As $\mathrm{d}\eta^\flat = 0$,
$\nabla_\eta\eta = \nabla(|\eta|^2/2)$, while $\Delta_H\eta = 0$ and so
$\Delta_\Def\eta = 2K\eta = -2\eta$. Hence $u = e^{-2\mu t}\eta$, with
pressure $p = -e^{-4\mu t}|\eta|^2/2$, solves the deformation equation;
by uniqueness it is the solution of Theorem~\ref{thm:weak-M} with initial
data $\eta$, and $\|u(t)\|^2 = e^{-4\mu t}\|\eta\|^2$, with equality in
the decay estimate. For the Hodge equation, by contrast, $\eta$ is a
stationary solution, with pressure $-|\eta|^2/2$: no decay estimate
holds for the Hodge equation on $H$, and the Hodge form, which vanishes
on $\eta$, is not coercive on $V$. \zw{The failure is topological and not
spectral. On $z$-independent functions, conjugation by
$(\cosh y)^{1/2}$ turns $-\Delta$ into
$-\partial_y^2 + \tfrac14 + \tfrac14\,\mathrm{sech}^2 y$ acting in
$L^2(\R,\mathrm{d}y)$, and dependence on $z$ only adds a nonnegative
term, so the bottom of the spectrum of the scalar Laplacian on
$\mathcal{C}$ is $1/4$, as on $\HH^2$.}

On a simply connected $M$ the rate improves: \eqref{eq:hodge-H} holds on
$V$ by density, with $\lambda_0\geq\kappa^2/4$ by McKean's
inequality~\cite{McK70}, so that~\eqref{eq:def-hodge-form} gives
$\langle-\Delta_\Def u,u\rangle\geq(\lambda_0+2\kappa^2)\|u\|^2$, and the
rate becomes $2\mu(\lambda_0+2\kappa^2)\geq\frac{9}{2}\mu\kappa^2$.}
\end{remark}

\subsection{The Korn inequality and the kernel of Def}
\label{sec:korn}

The natural energy functional for the deformation Laplacian is
$\|\Def\,u\|^2$, not $\|\nabla u\|^2$. The relationship between these
two norms is given by the \emph{Korn inequality}: for divergence-free
vector fields on a compact Riemannian manifold,
\begin{equation}\label{eq:korn}
\|\Def\,u\|^2 \geq C\,\|\nabla u\|^2 - \rev{C'}\,\|u\|^2,
\end{equation}
for constants $C > 0$ and \rev{$C'$} $\geq 0$ depending on the geometry (see
Duvaut and Lions~\cite{DL76} for the classical case and Chen and
Jost~\cite{CJ02} for the Riemannian \rev{generalization}). The kernel of
$\Def$ (where $\Def\,u = 0$) consists of the \emph{Killing fields}:
vector fields whose flow is an isometry. This is exactly the space of
rigid motions of the manifold, and the Korn inequality says that the
deformation Laplacian is coercive modulo Killing fields.

Physically, this is the correct \rev{behavior}. Viscosity dissipates
differential motion (shear) but cannot dissipate rigid motions (uniform
rotation, uniform translation), because rigid motions produce no strain.
An operator that dissipated Killing fields would be unphysical: it would
remove angular momentum from a uniformly rotating fluid without any
mechanism to absorb it.

\subsection[Long-time behavior on compact manifolds]{Long-time \rev{behavior} on compact manifolds}

Samavaki and Tuomela~\cite{ST20} investigated the long-time \rev{behavior} of
solutions to the time-dependent Navier-Stokes equation on compact
manifolds. Their key finding is:

\emph{With the deformation Laplacian}, solutions converge as
$t\to\infty$ to Killing fields. On $\Sph^2$, these correspond to rigid
rotations.

\emph{With the Hodge Laplacian}, solutions converge to harmonic 1-forms.
On $\Sph^2$, there are no non-trivial harmonic 1-forms, so solutions
converge to zero.

The deformation-Laplacian \rev{behavior} is physically correct: viscous
dissipation should remove differential rotation but preserve bulk
rotation, because viscosity cannot create or destroy angular momentum. A
uniformly rotating spherical fluid should persist indefinitely, not be
brought to rest. \zw{For an intrinsic fluid the} Hodge-Laplacian
\rev{behavior} is \zw{unphysical, since} a uniformly rotating sphere of
fluid \zw{that experiences} no external torque should not lose its
angular momentum to viscous dissipation. \zw{In the thin-shell setting
of Section~\ref{sec:shell} there is such a torque. The Hodge wall
condition gives $\partial_r U^i = -(2/r)\,U^i \neq 0$ on the walls, so
they exert a tangential stress, and the net torque on a rigidly rotating
layer
opposes the rotation. To leading order in $\varepsilon$ it removes the
rotation at the rate given by $\Delta_H\xi = -2\xi/R^2$ for a rotation
field $\xi$ on the sphere of radius $R$.}

The kernel of the deformation operator (Killing fields) contains the
rotational symmetries. The kernel of the Hodge operator (harmonic forms)
does not. The kinematic argument explains why: the deformation operator
responds to the strain rate $\Lie_u g$, which vanishes for Killing
fields by definition (Killing fields are exactly those whose flow
preserves the metric). The Hodge operator responds to
$\mathrm{d}u^\flat$ (the vorticity), which does not vanish for rotation.

%%%%%%%%%%%%%%%%%%%%%%%%%%%%%%%%%%%%%%%%%%%%%%%%%%
\section{Discussion}
\label{sec:discussion}
%%%%%%%%%%%%%%%%%%%%%%%%%%%%%%%%%%%%%%%%%%%%%%%%%%

We have presented a two-part account of the viscosity operator question
on Riemannian manifolds. For intrinsic fluids, the Lagrangian kinematic
construction uniquely selects the deformation Laplacian: the strain rate
is forced by the geometry of inner products to be $\Lie_u g$, leaving no
room for the antisymmetric (vorticity) contribution that would lead to
the Hodge Laplacian. For thin-shell fluids, the
decomposition~\eqref{eq:decomp} shows that the intrinsic part of the
ambient operator is always the deformation Laplacian, and the boundary
condition determines an additional extrinsic contribution. This
\rev{accounts for the difference between the results of} Temam-Ziane and
Miura\rev{; the analysis on arbitrary hypersurfaces is
in~\cite{WB-shell}}.

Several directions for further work are natural.

\emph{Noncompact manifolds and weak solutions.} The kinematic argument
operates at the level of smooth (classical) solutions. As \rev{shown}
in Section~\ref{sec:weak-solutions}, \rev{on noncompact surfaces with
curvature bounded above by a negative constant the
identity~\eqref{eq:coerce-def} takes the place of the Korn inequality,
without simple connectivity, and gives global weak solutions with
exponential decay for data in $H$. In all dimensions,
Pierfelice~\cite{Pie17} develops Stokes estimates and Fujita-Kato theory
for the deformation operator on simply connected manifolds \zw{of
nonpositive sectional curvature with Ricci curvature pinched between
two negative constants}, and~\cite{WB-H3} treats the $L^p$ spectrum of the
Stokes operator on hyperbolic 3-space and exponential decay for the
Navier-Stokes equations there. We have not attempted a
three-dimensional version of Theorem~\ref{thm:weak-M}; its uniqueness
argument is specific to two dimensions.}

\emph{General thin-shell geometries.} Our thin-shell decomposition is
carried out for the sphere\rev{. The extension to arbitrary
hypersurfaces is carried out formally, by matched asymptotics,
in~\cite{WB-shell}. There the intrinsic piece is again the deformation
Laplacian, and all extrinsic dependence is confined to the radial
boundary-shear piece; the vorticity-free and stress-free wall conditions
give the Hodge and deformation Laplacians, which differ by twice the
Ricci curvature whatever the extrinsic geometry; and a one-parameter
family of wall conditions joins them, with an effective operator that
couples to the extrinsic geometry, through the square of the shape
operator, only in between. The limits are made rigorous there, as Mosco
convergence of the viscous energy forms, on a class of surfaces of
revolution.}

\emph{Relativistic and dynamical-spacetime extensions.} In general
relativity, the kinematic construction carries over with the
four-velocity $u^\alpha$ and the spatial metric $h_{\alpha\beta} =
g_{\alpha\beta} + u_\alpha u_\beta$. The strain rate $\Lie_u h$
decomposes into the shear tensor $\sigma_{\alpha\beta}$ (symmetric
trace-free transverse part) and the expansion scalar $\theta$ (trace),
which are exactly the ingredients first used by Eckart~\cite{Eck40} and
subsequently adopted in Israel-Stewart~\cite{IS79} and
Bemfica-Disconzi-Noronha~\cite{BDN22} theories of relativistic
dissipative hydrodynamics. The kinematic argument \rev{is consistent
with these theories' use of} $\sigma_{\alpha\beta}$ and $\theta$ as the
sources of viscous stress, rather than the vorticity
$\omega_{\alpha\beta}$: the strain rate has no antisymmetric part by
construction. \rev{The relativistic version of the construction is
developed in~\cite{WB-rel}.}

%%%%%%%%%%%%%%%%%%%%%%%%%%%%%%%%%%%%%%%%%%%%%%%%%%
\appendix
\section{Tensor verification of the thin-shell decomposition}
\label{app:tensor}
%%%%%%%%%%%%%%%%%%%%%%%%%%%%%%%%%%%%%%%%%%%%%%%%%%

Here we verify the decomposition~\eqref{eq:decomp} by direct calculation
in spherical coordinates. The flat metric of $\R^3$ in coordinates
$(r,x^1,x^2)$ is
\begin{equation}
\mathrm{d}s^2 = \mathrm{d}r^2 + r^2\tilde{g}_{jk}\,\mathrm{d}x^j\mathrm{d}x^k,
\end{equation}
with inverse metric components $g^{rr} = 1$, $g^{jk}=r^{-2}\tilde{g}^{jk}$.
The non-vanishing Christoffel symbols are
\begin{equation}
\Gamma^r_{jk} = -r\tilde{g}_{jk},\quad \Gamma^i_{rj}
= \Gamma^i_{jr} = \frac{1}{r}\,\delta^i_j,\quad \Gamma^i_{jk}
= \tilde{\Gamma}^i_{jk},
\end{equation}
where $\tilde{\Gamma}^i_{jk}$ are the Christoffel symbols of $\Sph^2$.
Let $U$ be a purely tangential field ($U^r = 0$).

\paragraph{Radial trace.} We compute $\nabla_r U^i
= \partial_r U^i + \frac{1}{r}U^i$, then
\begin{align}
\nabla_r\nabla_r U^i &= \partial_r\!\left(\partial_r U^i
+ \frac{1}{r}U^i\right) + \frac{1}{r}\!\left(\partial_r U^i
+ \frac{1}{r}U^i\right)\nonumber\\
&= \partial_r^2 U^i + \frac{2}{r}\,\partial_r U^i.
\end{align}

\paragraph{Angular trace.} We find $\nabla_k U^i = \tilde{\nabla}_k U^i$
and $\nabla_k U^r = -r\tilde{g}_{kl}U^l$. The second covariant
derivative evaluates to
\begin{align}
\nabla_j\nabla_k U^i &= \tilde{\nabla}_j\tilde{\nabla}_k U^i
- \tilde{g}_{kl}U^l\delta^i_j + r\tilde{g}_{jk}\partial_r U^i
+ \tilde{g}_{jk}U^i.
\end{align}
Contracting with $g^{jk} = r^{-2}\tilde{g}^{jk}$ and using
$\tilde{g}^{jk}\tilde{g}_{jk} = 2$ (the dimension of $\Sph^2$):
\begin{equation}
g^{jk}\nabla_j\nabla_k U^i = \Delta_B^{(2)} U^i
+ \frac{2}{r}\,\partial_r U^i + \frac{1}{r^2}\,U^i.
\end{equation}

\paragraph{Total.} Summing the radial and angular traces:
\begin{equation}
(\Delta_B^{(3)} U)^i = \Delta_B^{(2)} U^i + \frac{1}{r^2}\,U^i
+ \partial_r^2 U^i + \frac{4}{r}\,\partial_r U^i.
\end{equation}
Since $\Ric^{(2)}(U)^i = \frac{1}{r^2}U^i$ on a sphere of radius
$r$, evaluating at $r = R$ gives~\eqref{eq:decomp}.

\paragraph{Navier-slip evaluation.} The normal shear is $\sigma_{ri} =
\frac{1}{2}(\nabla_r U_i + \nabla_i U_r)$. Using $U_i =
r^2\tilde{g}_{ij}U^j$ and $U_r = 0$, this reduces to $\sigma_{ri} =
\frac{1}{2}r^2\tilde{g}_{ij}\partial_r U^j$. Zero normal shear forces
$\partial_r U^j = 0$ \rev{on the wall. Imposed on both walls, this gives
$\partial_r U^j = \partial_r^2 U^j = 0$ at $r = R$ to leading order in
$\varepsilon$ (Section~\ref{sec:shell})}.

\paragraph{Hodge-condition evaluation.} \rev{On a wall $r = r_\pm$,}
vanishing tangential vorticity requires $\partial_r U_i \rev{{}=
\partial_i U_r} = 0$\rev{, the last equality because $U_r = 0$ on the
wall}. Since $U_i = r^2\tilde{g}_{ij}U^j$\rev{, this is
$$
\partial_r U^i = -\frac{2}{r_\pm}\,U^i \qquad \text{at } r = r_\pm,
$$
and imposing it on both walls gives, as in Section~\ref{sec:shell},
$\partial_r U^i = -(2/R)\,U^i$ and $\partial_r^2 U^i = (6/R^2)\,U^i$ at
$r = R$, to leading order in $\varepsilon$. Then}
$\mathcal{F}_\text{rad} = \frac{6}{R^2}U^i - \frac{8}{R^2}U^i
= -\frac{2}{R^2}U^i = -2\,\Ric^{(2)}(U)^i$.

%%%%%%%%%%%%%%%%%%%%%%%%%%%%%%%%%%%%%%%%%%%%%%%%%%
\section*{Declarations}

\textbf{Conflict of Interest:} The authors declare that they have no conflict of interest.\\
\textbf{Data Availability:} The authors did not use any data during the analysis of the paper.

%%%%%%%%%%%%%%%%%%%%%%%%%%%%%%%%%%%%%%%%%%%%%%%%%%

\end{document}